\documentclass[letterpaper, 10 pt, conference]{ieeeconf}

\IEEEoverridecommandlockouts                              

\usepackage{graphics} 
\usepackage{graphicx}
\usepackage{epsfig} 
\usepackage{comment}
\usepackage{caption}
\usepackage{subcaption}
\usepackage{cite}
\usepackage{amsmath} 
\usepackage{amssymb}  
\usepackage{array}

\usepackage{xcolor}
\usepackage[colorlinks=true]{hyperref}

\let\labelindent\relax
\usepackage{enumitem}

\newtheorem{theorem}{Theorem}
\newtheorem{lemma}[theorem]{Lemma}

\newtheorem{xmpl}{Example}
\newtheorem{remark}{Remark}

\usepackage{algorithm}
\usepackage{algorithmic}

\title{\LARGE \bf A New Finite-Horizon Dynamic Programming Analysis of Nonanticipative Rate-Distortion Function for Markov Sources}

\author{Zixuan He$^{1}$, Charalambos D. Charalambous$^2$, and Photios A. Stavrou$^{1}$
 \thanks{*The work of Z. He was supported by the Huawei France-EURECOM Chair on Future Wireless Networks. The work of P. A. Stavrou was supported in part by the SNS JU project 6G-GOALS \cite{strinati:2024} under the EU’s Horizon programme (Grant Agreement No. 101139232) and by the Huawei France-EURECOM Chair on Future Wireless Networks.}
	\thanks{$^{1}$Z. He and P. A. Stavrou are with the Foundation and Algorithm Group, Communication Systems Department, EURECOM, France. {\tt\small email: \{zixuan.he,fotios.stavrou\}@eurecom.fr}}
	\thanks{$^{2}$C. D. Charalambous is with the Department of Electrical and Computer Engineering, University of Cyprus, Cyprus.
		{\tt\small email: chadcha@ucy.ac.cy}}%
}

\begin{document}

\maketitle
\thispagestyle{empty}
\pagestyle{empty}

\begin{abstract}
This paper deals with the computation of a non-asymptotic lower bound by means of the nonanticipative rate-distortion function (NRDF) on the discrete-time zero-delay variable-rate lossy compression problem for discrete Markov sources with per-stage, single-letter distortion. First, we derive a new information structure of the NRDF for Markov sources and single-letter distortions. Second, we derive new convexity results on the NRDF, which facilitate the use of Lagrange duality theorem to cast the problem as an unconstrained partially observable finite-time horizon stochastic dynamic programming (DP) algorithm subject to a probabilistic state (belief state) that summarizes the past information about the reproduction symbols and takes values in a continuous state space. Instead of approximating the DP algorithm directly, we use Karush-Kuhn-Tucker (KKT) conditions to find an implicit closed-form expression of the optimal control policy of the stochastic DP (i.e., the minimizing distribution of the NRDF) and approximate the control policy and the cost-to-go function (a function of the rate) stage-wise, via a novel dynamic alternating minimization (AM) approach, that is realized by an offline algorithm operating using backward recursions, with provable convergence guarantees. We obtain the clean values of the aforementioned quantities using an online (forward) algorithm operating for any finite-time horizon. Our methodology provides an approximate solution to the exact NRDF solution, which becomes near-optimal as the search space of the belief state becomes sufficiently large at each time stage. We corroborate our theoretical findings with simulation studies where we apply our algorithms assuming time-varying and time-invariant binary Markov processes.
	
\end{abstract}

\section{Introduction}
\label{section:introduction}

Classical lossy source coding problem refers to the scenario where one encodes a long block of source symbols so that it allows the distortion asymptotically to achieve the Shannon's limit at the minimum empirical bit-rate \cite{shannon:1959}. The long block codes that are required for the compression scheme to operate on the Shannon's limit induce long coding delays, and this makes classical lossy compression undesirable in many emerging delay-sensitive applications such as networked control systems \cite{hespanha:2007}, wireless sensor networks \cite{akyildiz:2002}, and more recently in the merging field of semantic communications \cite{strinati:2024}.
\par A more suitable lossy compression paradigm to deal with delay-constrained applications, is the so-called zero-delay lossy source coding problem. Contrary to classical lossy compression, the compressed symbols of the information source symbols are generated by an encoder  
without delay, and are communicated over a discrete noiseless channel to the decoder, which reconstructs the source symbols, also without
delay, subject to a fidelity. The discrete noiseless channel may operate assuming either fixed or variable-rates.

{\it Literature Review:} A number of important results are documented in the literature on zero-delay lossy compression schemes. The early works \cite{witsenhausen:1979,varaiya:1983} laid the foundations to understand the structural properties of the optimal zero-delay codes (assuming primarily fixed-rate) using stochastic control and dynamic programming (DP) \cite{kumar:1986,bertsekas:2005}. Similar results were generalized by many researchers, see e.g., \cite{teneketzis:2006,mahajan:2009,linder:2014,wood:2017,ghomi:2021}. In \cite{kaspi:2012}, the authors considered structural theorems for the zero-delay lossy source coding problem with variable-rate constraints with and without side information at the decoder. Recently in \cite{cregg:2024}, the authors invoked reinforcement learning algorithms via a quantized Q-learning method in the infinite-time horizon to approximate near-optimally, the zero-delay coding problem for fixed-rates assuming finite-alphabet stationary Markov sources. On another relevant research direction, instead of attacking directly the zero-delay lossy compression problem subject to variable-rate constraints, information-theoretic upper and lower bounds are derived building on the earlier work of \cite{gorbunov:1972} who introduced an information-theoretic measure called nonanticipatory $\epsilon-$entropy\footnote{Also found in the literature as sequential rate-distortion-function (RDF) \cite{tatikonda:2004} and nonanticipative RDF (NRDF) \cite{charalambous:2014}.}. This line of research primarily studied bounding techniques on linear (perhaps controlled) Markov systems driven by additive Gaussian or non-Gaussian noise and also, the derivation of information-theoretic closed-form solutions \cite{tatikonda:2004,silva:2011,derpich:2012,tanaka:2017,stavrou:2018siam, kostina:2019a, charalambous:2022,stavrou:2024tac}.

{\it Contributions:} In this paper, we analyze a non-asymptotic lower bound on the empirical rates of a discrete-time zero-delay variable-rate lossy source coding system assuming discrete and possibly time-varying Markov sources subject to a per-stage, average single-letter distortion criterion. First, we derive a structural property of our lower bound (obtained via NRDF) (see Lemma \ref{lemma:information-structure}) for the specific class of sources and fidelity constraint. Second, we derive some new convexity properties (under certain conditions) that characterize the functionals of the resulting optimization problem (see Theorems \ref{theorem:sequential-convexity}, \ref{theorem:convexity-set}), and we leverage those to cast it as an unconstrained partially observable finite-time horizon stochastic DP algorithm with a continuous state space \cite{bertsekas:2005} (see equations \eqref{DP-time-n}, \eqref{DP-time-t}). Instead of solving the DP algorithm directly, we optimize with respect to the control policy (which in our case corresponds to the minimizing distribution of the NRDF) that results in some implicit closed-form recursions obtained backward in time. These are computed by proposing a new dynamic AM scheme (see Lemma \ref{lemma:double-min}) that approximates offline (see Algorithm \ref{algo:DP-backward}) the control policy and the cost-to-go function (a function of the rate) by discretizing the continuous state space into a finite state space at each time stage. Subsequently, we propose a forward (online) algorithm (see Algorithm \ref{algo:DP-forward}) to compute the clean values of the aforementioned quantities for any finite-time horizon. Our offline scheme has provable convergence guarantees per stage (see Theorems \ref{theorem:convergence}, \ref{theorem:stopping-criterion}) and approaches a near-optimal solution once the search space of the discretized (finite) state becomes sufficiently large. We corroborate our theoretical results with numerical simulations in which we demonstrate the behavior of both time-varying and time-invariant binary Markov sources and their corresponding control policy and the cost-to-go (i.e., the rate per stage). An interesting aspect of implementing our offline algorithm in our computational studies, is the use of parallel processing, which alleviates the issue of computational complexity compared to standard single-thread processing. To the best of our knowledge, this is the first paper in which (i) the optimization of NRDF for discrete Markov sources and single-letter distortion is reformulated as an unconstrained partially observable finite-time horizon stochastic DP algorithm with continuous state; (ii) the control policy is approximated by means of a novel dynamic AM scheme realized via an offline training algorithm that generalizes the known Blahut-Arimoto algorithm (BAA) \cite{blahut:1982}, followed by an online computation. 

\textit{Notation}: $\mathbb{N}\triangleq\{1,2,\ldots\}$, $\mathbb{N}_0\triangleq\{0,1,\ldots\}$, 
and $\mathbb{N}_j^n\triangleq\{j,\ldots,n\},~j\leq{n},~n\in\mathbb{N}$. We denote a sequence of random variables (RVs) by $X^t=\{X_0,X_1,\ldots,X_t\},t\in\mathbb{N}_0^n$ and their values by $x^t\in\mathcal{X}^t=\{\mathcal{X}_0,\ldots,\mathcal{X}_t\}$, where $\mathcal{X}_t$ denotes the alphabet and hence $\mathcal{X}^t$ the alphabet sequence. A truncated sequence of RVs is denoted by $X_{j}^t=\{X_j,\ldots,X_t\}, j\in\mathbb{N}_0^t, {t}\geq{j}$, and its realizations by $x_{j}^t=\{x_j,\ldots,x_t\}\in\mathcal{X}_{j}^t=\{\mathcal{X}_{j},\ldots,\mathcal{X}_t\}$,~${t}\geq{j}$. The distribution of a RV $X$ on $\mathcal{X}$ is denoted by $P(x)$ and the conditional distribution of a RV $Y$ given $X=x$ is denoted by $P(y|x)$. We indicate with square brackets the functional dependency between mathematical objects, e.g. $P[Q](x)$ and $P[y](x)$ express the functional dependence of a distribution $P(x)$ on another distribution $Q$ and on another realization $y$, respectively. We denote by $\mathbb{E}\{\cdot\}$ the expectation operator, and by $\mathbb{E}^{P^o}\{\cdot\}$ the expectation with respect to a given distribution $P(\cdot)=P^o(\cdot)$.

\section{Problem Statement and Preliminaries}
\label{section:problem-statement}
We consider the discrete-time zero-delay lossy source coding system illustrated in Fig. \ref{fig:l-sc}, operating at any finite-time horizon $n\in\mathbb{N}_0$. The operation of the system is described as follows.
\begin{figure}
\centering\includegraphics[width=.45\textwidth]{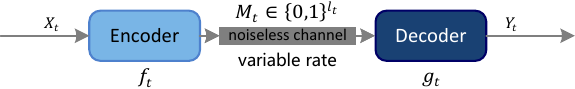}
\caption{A generic zero-delay lossy source coding system}
\label{fig:l-sc}
\end{figure}
\paragraph*{Operation} At each time instant $t\in\mathbb{N}_1^n$, the \textit{source} is modeled as a Markov process (not necessarily time-homogeneous), i.e., with transition probability distribution $P_t(x_t|x_{t-1})$ and initial distribution $P_0(x_0)$, which induce the joint distribution of the sequence of RVs $X^t, t\in\mathbb{N}_0^n$. For any $X_t\in{\cal X}_t$, we assume that the cardinality of  $\mathcal{X}_t$ is finite. The \textit{encoder (E)} $f_t$ encodes the information $X_t$ generated from the source based on the past information $X^{t-1}$ and produces a variable-rate codeword $M_t\in{\cal M}_t\subset\{0,1\}^{l_t}$ of length $l_t$ and expected rate $R_t=\mathbb{E}|l_t|$. The \textit{decoder (D)} $g_t$ receives the past and current codewords $M^t$ to reproduce $Y_t\in\mathcal{Y}_t$,  provided that $Y^{t-1}$ is already reproduced. Again, we assume that for any $Y_t\in{\cal Y}_t$, the cardinality of  $\mathcal{Y}_t$ is finite.  Formally, the \textit{(E), (D)} pair is specified by the sequence of possibly stochastic mappings $f_t:\mathcal{M}^{t-1}\times\mathcal{X}^t\rightarrow\mathcal{M}_t$, and $g_t:\mathcal{M}^t\to\mathcal{Y}_t$, respectively. We note that at $t=0$, the encoder output is $m_0=f_0(x_0)$ and the decoder's output $y_0=g_0(m_0)$. This means that no prior information is assumed at both the encoder and the decoder.

\paragraph*{Fidelity constraint} The system in Fig. \ref{fig:l-sc}, is penalized at each time instant by an additive fidelity $\{D_t\geq{0}:~t\in\mathbb{N}_0^n\}$, the constraint between the source process $\{X_t:t\in\mathbb{N}_0^n\}$ and the reproduction process $\{Y_t:t\in\mathbb{N}_0^n\}$ is described by the per-stage, average single-letter distortion, i.e., $\mathbb{E}\{ \rho_t(X_t,Y_t) \} \leq D_t$, where $\rho_t(x_t,y_t)$ is the single-letter distortion function between the information source $x_t$ and its reproduction $y_t$ at each $t$.

\paragraph*{Performance} The system model in Fig. \ref{fig:l-sc}, subject to the fidelity criterion described above, can be computed by the following empirical rate optimization:
\begin{align}
    R_{[0,n]}^o(D_0,\ldots,D_n) \triangleq
    \inf_{\substack{\{f_t,\ g_t:\ t\in\mathbb{N}_0^n\} \\ \mathbb{E}\{\rho_t(X_t,Y_t)\}\leq D_t,  \forall t\in\mathbb{N}_0^n}} \frac{1}{n+1}\sum_{t=0}^n{R}_t.
\label{RDF-operational}
\end{align}
The solution of \eqref{RDF-operational} depends on every possible code, which makes it very challenging to develop a globally optimal solution achieving also a reasonable computational complexity. As a result, it makes sense to pursue approximate solutions via bounds.

A well-known lower bound on \eqref{RDF-operational} assuming a Markov source and a per-stage, average single-letter distortion, is the NRDF \cite{gorbunov:1972} (see also \cite{stavrou:2020entropy}) given as follows:
\begin{align}
    R_{[0,n]}^{na}(D_0,\ldots,D_n) \triangleq \inf_{\mathcal{Q}_{[0,n]}(D_0,\ldots,D_n)} \frac{1}{n+1} I(X^n\rightarrow Y^n)
\label{NRDF}
\end{align}
where the constraint set 
\begin{align}
\mathcal{Q}_{[0,n]}&(D_0,\ldots,D_n)\triangleq\{{P}_t(y_t|y^{t-1},x_t):\nonumber\\
&~\mathbb{E}\{\rho_t(X_t,Y_t)\}\leq{D_t},~ {\forall t\in\mathbb{N}_0^n}\}\label{initial_distortion_constraint_set}
\end{align}
and $I(X^n\rightarrow Y^n)$ is a variant of directed information (DI) \cite{massey:1990} defined as follows
\begin{align}
    I(X^n&\rightarrow Y^n)\triangleq\sum_{t=0}^n \mathbb{E} \left\{\log \left(\frac{P_t(Y_t|Y^{t-1},X_t)}{P_t(Y_t|Y^{t-1})}\right)\right\}\nonumber\\
    &= \sum_{t=0}^n I(X_t;Y_t|Y^{t-1})\nonumber\\
    &=\sum_{t=0}^n \int_{\mathcal{X}_{t-1}^{t}\times\mathcal{Y}^{t}}\log\left(\frac{P_t(y_t|y^{t-1},x_{t})}{P_t(y_t|y^{t-1})}\right) P_t(x_t|x_{t-1}) 
    \nonumber\\
    &\qquad P_t(y_t|y^{t-1},x_{t}){P}_{t}(x_{t-1}|y^{t-1})P_{t}(y^{t-1}) \nonumber\\
    &\stackrel{(a)}\equiv\mathbb{I}_{[0,n]}(\overleftarrow{P}_{[0,n]}(x^n),\overrightarrow{P}_{[0,n]}(y^n|x^n))
    \label{DI-func-dependence}
\end{align}
where $(a)$ denotes the functional dependence of $I(X^n\rightarrow Y^n)$ with respect to (wrt) $\overleftarrow{P}_{[0,n]}(x^n)\triangleq\otimes_{t=0}^nP_t(x_t|x_{t-1})$ and $\overrightarrow{P}_{[0,n]}(y^n|x^n)\triangleq\otimes_{t=0}^n{P}_t(y_t|y^{t-1},x_t)$, respectively.

We state some noteworthy remarks related to \eqref{NRDF}.

\begin{remark} (On the bound of \eqref{NRDF})
\label{remark:1}
{\bf (i)} In accordance with the system model in Fig. \ref{fig:l-sc}, \eqref{NRDF} does not assume any prior information at $t=0$ available in the system model but only $P_0(x_0)$ and $P_0(y_0)$ are fixed. This means that at $t=0$, $I(X_0;Y_0|Y^{-1})\equiv{I(X_0;Y_0)}$. {\bf (ii)} The bound of \eqref{NRDF} has been thoroughly studied for continuous sources, see e.g.,  \cite{tanaka:2017,stavrou:2018siam,kostina:2019a,charalambous:2022} but not adequately for discrete sources, with a notable exception perhaps the work of \cite{stavrou:2015} which alas does not provide general methodologies for computation or tangible proofs to certain analytical expressions. {\bf (iii)} Although \eqref{NRDF} under the constraint set \eqref{initial_distortion_constraint_set} is expected to be a convex program, it is not. Specifically, \eqref{NRDF} is a convex program when $I(X^n\rightarrow{Y^n})$ is a convex function wrt the product $\overrightarrow{P}_{[0,n]}(y^n|x^n)\triangleq\otimes_{t=0}^n P_t(y_t|y^{t-1},x_t)$ for a fixed $\overleftarrow{P}_{[0,n]}(x^n)\triangleq\otimes_{t=0}^n P_t(x_t|x_{t-1})$ for a constraint set defined as 
\begin{align}
&\overrightarrow{\mathcal{Q}}_{[0,n]}(D_0,\ldots,D_n)\triangleq\{\overrightarrow{P}(y^n|x^n)\triangleq\otimes_{t=0}^nP_t(y_t|y^{t-1},x_t):\nonumber\\
&\mathbb{E}\{\rho_t(X_t,Y_t)\}\leq{D}_t,~\forall{t\in\mathbb{N}_0^n}\}\label{distortion_set_product_space}
\end{align}
(see e.g., \cite{charalambous:2016}). Hence, {\it no proof of the convexity of \eqref{NRDF} exists} wrt the constraint set \eqref{initial_distortion_constraint_set}.  {\bf (iv)} There is no generic implicit or explicit solution of \eqref{NRDF} under the constraint set of \eqref{initial_distortion_constraint_set}. However, implicit closed-form recursions of the optimal minimizer are reported (without a complete proof) in \cite[Theorem 4.1]{stavrou:2018siam} assuming the constraint set in \eqref{distortion_set_product_space}.
\end{remark} 

In this work, we provide a generic methodology to approximate \eqref{NRDF} assuming discrete Markov sources aiming primarily to close the gaps mentioned in Remark \ref{remark:1}, {\bf (iii)}, {\bf (iv)}.

\section{Main Results}
\label{section:main-results}

In this section, we give our main results. To do it, we restrict ourselves to finite alphabet spaces, e.g., with cardinality $|\mathcal{X}_t|<\infty$, $|\mathcal{Y}_t|<\infty,~\forall{t}$, throughout the paper.

First, we give a new information structure that simplifies the multi-letter variant of DI in \eqref{NRDF}.
\begin{lemma}{(Structural Property)}
\label{lemma:information-structure}
For a given Markov source $\{P_t(x_t|x_{t-1}):~t\in\mathbb{N}_0^n\}$ and a single letter distortion function $\{\rho_t(x_t,y_t):~t\in\mathbb{N}_0^n\}$, the characterization in \eqref{NRDF} can be simplified as follows\footnote{For non-empty finite sets, we can replace infimum with minimum due to the compactness of the constraint sets.}
\begin{align}
    R_{[0,n]}^{na}(D_0,\ldots,D_n) \triangleq \min_{\tilde{\mathcal{Q}}_{[0,n]}(D_0,\ldots,D_n)} I(X^n\rightarrow Y^n),
\label{NRDF-simplified}
\end{align}
where 
\begin{align}
    \tilde{\mathcal{Q}}_{[0,n]}(D_0,&\ldots,D_n)\triangleq \Big\{P_t(y_t|y_{t-1},x_t): \nonumber\\
    &\mathbb{E}\{\rho_t(X_t,Y_t)\} \leq D_t,~\forall t\in\mathbb{N}_0^n\Big\}
    \label{DistortionSet-new}\\
    I(X^n\rightarrow Y^n) \triangleq& \sum_{t=0}^n I(X_t;Y_t|Y_{t-1}) \label{CondMI-3-letter}.
\end{align}
\end{lemma}
\begin{proof}
We outline the proof due to space limitations. The goal is to show that for the specific class of sources and single-letter distortion, then via \eqref{NRDF} we can obtain a lower bound which is achievable. To obtain the lower bound, we make use of a sequential version of variational equalities of DI derived in \cite[Theorem 19, {\bf Part B}, {\bf (ii)}]{charalambous:2016} and then further minimize with respect to the constraint set in \eqref{DistortionSet-new}. The upper bound can be trivially obtained upon observing that ${\cal Q}_{[0,n]}(D_0,\ldots,D_n)\subset\tilde{\cal Q}_{[0,n]}(D_0,\ldots,D_n)$ in \eqref{NRDF} which also implies \eqref{CondMI-3-letter}.
\end{proof}
The result in Lemma \ref{lemma:information-structure} is generic and holds for both discrete and continuous alphabets. Indeed, such property is already verified for jointly Gaussian-Markov processes, e.g., \cite{stavrou:2018siam,gorbunov:1972b}.

The next two results provide conditions to ensure new convexity properties for the expression in \eqref{NRDF-simplified}. The first result demonstrates a new convexity property of \eqref{CondMI-3-letter} for a given posterior distribution $\{P_t(x_{t-1}|y_{t-1})\equiv{P}_t^o(x_{t-1}|y_{t-1}):~t\in\mathbb{N}_0^n\}$, wrt the sequence of minimizing distributions $\{P_t(y_t|y_{t-1},x_t):~t\in\mathbb{N}_0^n\}$.

\begin{theorem}{(Convexity of \eqref{CondMI-3-letter})}
\label{theorem:sequential-convexity}
For a fixed source distribution $\{P_t(x_t|x_{t-1}):~t\in\mathbb{N}_0^n\}$, and a given posterior distribution $\{P_t(x_{t-1}|y_{t-1})\equiv{P}_t^o(x_{t-1}|y_{t-1}):~t\in\mathbb{N}_0^n\}$ obtained for a fixed $Y_{t-1}=y_{t-1}$, define the conditional mutual information $I(X_t;Y_t|Y_{t-1}=y_{t-1})$ as follows
\begin{align}
    I(X_t;&Y_t|Y_{t-1}=y_{t-1})
    \triangleq \sum_{x_{t-1}\in\mathcal{X}_{t-1}}\bigg(\sum_{x_t\in\mathcal{X}_t,y_t\in\mathcal{Y}_t}\nonumber
    \\ &\log\left(\frac{P_t(y_t|y_{t-1},x_{t})}{P_t(y_t|y_{t-1})}\right)P_t(x_t|x_{t-1}) P_t(y_t|y_{t-1},x_{t})\bigg) \nonumber\\ & P^o_{t}(x_{t-1}|y_{t-1}),~\forall{t\in\mathbb{N}_0^n}. 
\label{CondMI-stagewise}
\end{align}
Then, \eqref{CondMI-stagewise} is a convex functional wrt $\{P_t(y_t|y_{t-1},x_{t}):~t\in\mathbb{N}_0^n\}$. Moreover,  the additive term
\begin{align}
    I(X^n\rightarrow Y^n)=\sum_{t=0}^n\sum_{\substack{y_{t-1}\\
    \in\mathcal{Y}_{t-1}}}P_t(y_{t-1})I(X_t;Y_t|Y_{t-1}=y_{t-1})
\label{CondMI-fixed-y_pre}
\end{align}
is also convex wrt $\{P_t(y_t|y_{t-1},x_{t}):t\in\mathbb{N}_0^n\}$.
\end{theorem}
\begin{proof}
We outline the proof due to space limitations. Under the conditions of the theorem, we can show convexity of \eqref{CondMI-stagewise} wrt $\{P_t(y_t|y_{t-1},x_{t}):t\in\mathbb{N}_0^n\}$ at each $t$ by applying repetitively the log-sum inequality \cite{cover-thomas:2006}. To show convexity of \eqref{CondMI-3-letter} wrt $\{P_t(y_t|y_{t-1},x_{t}):t\in\mathbb{N}_0^n\}$ we take the sum of all \eqref{CondMI-stagewise} averaged wrt the non-negative $P_t(y_{t-1})$ and establish the result since the convexity-preserving rule \cite{boyd:2004} holds.
\end{proof}

The following result establishes the convexity of the constraint set in \eqref{DistortionSet-new} for a given posterior distribution $\{P_t(x_{t-1}|y_{t-1})\equiv{P}_t^o(x_{t-1}|y_{t-1}):~t\in\mathbb{N}_0^n\}$.

\begin{theorem}{(Convexity of \eqref{DistortionSet-new})}
\label{theorem:convexity-set}
For a fixed source distribution $\{P_t(x_t|x_{t-1}):~t\in\mathbb{N}_0^n\}$, and a given posterior distribution $\{P_t(x_{t-1}|y_{t-1})\equiv{P}_t^o(x_{t-1}|y_{t-1}):~t\in\mathbb{N}_0^n\}$ obtained for a fixed $Y_{t-1}=y_{t-1}$, define the constraint set
\begin{align}
    &\tilde{Q}_{[0,n]}(D_0,D_1[y_{0}, P_1^o],\ldots,D_n[y_{n-1}, P_n^o])\triangleq\nonumber\\
    &\Big\{P_t(y_t|y_{t-1},x_t), t\in\mathbb{N}_0^n:~\mathbb{E}\Big\{\rho_0(X_0,Y_0)\Big\}\leq D_0;\nonumber\\
    &\mathbb{E}^{P_t^o}\Big\{\rho_t(X_t,Y_t){\Big|} 
    Y_{t-1}=y_{t-1}\Big\} \leq D_t[y_{t-1},P_t^o],~t\in\mathbb{N}_1^n\Big\}, 
\label{DistortionSet-fixed-parameters}  
\end{align}
where $D_t[y_{t-1},P_t^o]$ denotes the functional dependence of $D_t$ wrt the given $\{P_t(x_{t-1}|y_{t-1})\equiv{P}_t^o(x_{t-1}|y_{t-1}):~t\in\mathbb{N}_0^n\}$ obtained for a fixed $Y_{t-1}=y_{t-1}$. Then, \eqref{DistortionSet-fixed-parameters} forms a convex set. Moreover, if we average the constraints in \eqref{DistortionSet-fixed-parameters} wrt $y_{t-1}\in{\mathcal{Y}}_{t-1}$ for each $t\in\mathbb{N}_1^n$,
\eqref{DistortionSet-new} is also convex for a given $\{P_t(x_{t-1}|y_{t-1}) \equiv P_t^o(x_{t-1}|y_{t-1}),~t\in\mathbb{N}_0^n\}$.  
\end{theorem}
\begin{proof}
We sketch the proof due to space limitations. The conditions of the theorem ensure that \eqref{DistortionSet-fixed-parameters} is the union of multiple disjoint sets, i.e., $\tilde{Q}_0(D_0)=\{P(y_0|x_0):\mathbb{E}\{\rho_0(X_0,Y_0)\}\leq D_0\}$ and  $\tilde{Q}_t(D_t[y_{t-1},P_t^o])=\{P_t(y_t|y_{t-1},x_t):~\mathbb{E}^{P_t^o}\{\rho_t(X_t,Y_t)|Y_{t-1}=y_{t-1}\}\leq D_t[y_{t-1},P_t^o]\},~\forall{t\in\mathbb{N}_1^n}$. The convexity of each disjoint set is shown by proving that two elements in the set form a nonnegative linear combination. The convexity of \eqref{DistortionSet-new} for a given $\{P_t(x_{t-1}|y_{t-1}) \equiv P_t^o(x_{t-1}|y_{t-1}),~t\in\mathbb{N}_0^n\}$ can be proved by invoking again the convexity-preserving rule \cite{boyd:2004}. Particularly, averaging each convex disjoint set with respect to $P_t(y_{t-1}) > 0$ and summing the convex disjoint sets will preserve convexity.
\end{proof}
Using Theorems \ref{theorem:sequential-convexity}, \ref{theorem:convexity-set}, we can exploit Lagrange duality theorem \cite{boyd:2004} with Lagrange multipliers $\{s_t\leq0,~t\in\mathbb{N}_0^n\}$, to cast \eqref{NRDF-simplified} for a given $\{P_t(x_{t-1}|y_{t-1})\equiv P_t^o(x_{t-1}|y_{t-1}),~t\in\mathbb{N}_0^n\}$ obtained for a fixed $Y_{t-1}=y_{t-1}$, into the following unconstrained convex optimization problem:
\begin{align}
    &R^{un}_{[0,n]}(D_0,D_1[y_{0},P_0^o],\ldots,D_n[y_{n-1},P_n^o]) =\sup_{\{s_t\leq0:~t\in\mathbb{N}_0^n\}}\nonumber\\
    &\min_{\{P_t(y_t|y_{t-1},x_t):~t\in\mathbb{N}_0^n\}} \sum_{t=0}^n \mathbb{E}^{P_t^o} \Bigg\{ \log \left( \frac{P_t(Y_t|Y_{t-1},X_t)}{P_t(Y_t|Y_{t-1})}\right)\nonumber\\
    &\qquad\qquad-s_t \big(\rho_t(X_t,Y_t)-D_t[y_{t-1},P_t^o]\big)\Bigg{|}Y_{t-1}=y_{t-1} \Bigg\}.
\label{NRDF-Lagrange-dual_fixed_parameters}
\end{align}
Note that if in \eqref{NRDF-Lagrange-dual_fixed_parameters}, we average wrt to $P_t(y_{t-1})>0$, then, we can approximate \eqref{NRDF-simplified} for a given $\{P_t(x_{t-1}|y_{t-1})\equiv P_t^o(x_{t-1}|y_{t-1}),~t\in\mathbb{N}_0^n\}$. Moreover, if our search space is sufficiently large that can cover the vast majority of possible values of $\{P_t(x_{t-1}|y_{t-1})\equiv P_t^o(x_{t-1}|y_{t-1}),~t\in\mathbb{N}_0^n\}$, we can theoretically approach near-optimally \eqref{NRDF-simplified}.
\par Note that \eqref{NRDF-Lagrange-dual_fixed_parameters} can be reformulated using stochastic optimal control arguments as an {\it unconstrained partially observable finite-time horizon stochastic DP algorithm with a continuous state space} \cite{bertsekas:2005}. Particularly, the state of the stochastic optimal control problem is a generalization of the classical state called {\it belief} or {\it information state} and is given by the probabilistic state $\{P_t(x_{t-1}|y_{t-1}):~t\in\mathbb{N}_0^n\}$ (assuming $Y_{t-1}=y_{t-1},~\forall{t}$) with elements that explore all possible values between the interval $(0,1)$. Therefore, there are infinitely many choices for the elements of the probabilistic state hence the term continuous state space. To formalize this in an example, suppose that we are given ${\cal X}_t={\cal Y}_t={\cal X}=\{0,1\},~\forall{t}$. Then, a possible choice of the belief state $P_t(x_{t-1}|y_{t-1})$ would be to create a $2\times{2}$ stochastic matrix shown in Table \ref{table:belief-matrix},
where $(\alpha_t^{00},\alpha_t^{01})$ for each $t$, should explore all possible values in the interval $(0,1)$. 
The \textit{feedback control law or policy} is the minimizing distributions $\{P_t(y_t|y_{t-1},x_t):~t\in\mathbb{N}_0^n\}$, the \textit{random disturbance} refers to the source distributions $\{P_t(x_t|x_{t-1}):~t\in\mathbb{N}_0^n\}$ whereas the \textit{cost function} is the quantity $\log\left( \frac{P_t(y_t|y_{t-1},x_t)}{P_t(y_t|y_{t-1})}\right)-s_t\rho_t(x_t,y_t)$. A summary of the above in the context of stochastic optimal control problem with their one-to-one correspondence in \eqref{NRDF-Lagrange-dual_fixed_parameters} is illustrated in Table \ref{table:reformulated-variables}.
\begin{table}
\centering
\caption{$P_t(x_{t-1}|y_{t-1})$}
\begin{tabular}{c|cc}
    \hline
     & $y_{t-1} = 0$ & $y_{t-1} = 1$\\\hline
     $x_{t-1} = 0$ & $1 - \alpha_t^{00}$ & $1 - \alpha_t^{01}$\\
    $ x_{t-1} = 1 $ & $\alpha_t^{00}$ & $\alpha_t^{01}$ \\\hline
    \end{tabular}
\label{table:belief-matrix}
\end{table}
\begin{table}
\centering
\caption{Stochastic optimal control problem}
\begin{tabular}{lcc}
    \hline
    Variables of DP&
    Connection to \eqref{NRDF-Lagrange-dual_fixed_parameters}\\  \hline
    belief or information state & $P_t(x_{t-1}|y_{t-1})$ \\
    disturbance & $P_t(x_t|x_{t-1})$ \\
    control policy & $P_t(y_t|y_{t-1},x_t)$ \\
    cost function & $\log\left( \frac{P_t(y_t|y_{t-1},x_t)}{P_t(y_t|y_{t-1})}\right)-s_t\rho_t(x_t,y_t)$ \\
    \hline
\label{table:reformulated-variables}
\end{tabular}
\end{table}

In the sequel, we let $R_t(\cdot)$ denote the optimal expected cost or pay-off in \eqref{NRDF-Lagrange-dual_fixed_parameters} on the future time horizon $\{t,t+1,\ldots,n\}$. Then, for a given belief state $P_t(x_{t-1}|y_{t-1})=P_t^o(x_{t-1}|y_{t-1})$ obtained for a fixed $Y_{t-1}=y_{t-1}$, this quantity is described as follows
\begin{align*}
    &R_t(D_t[y_{t-1},P_t^o]) \\&= \min_{\substack{\{P_i(y_i|y_{i-1},x_i):\\i\in\{t,t+1,\ldots,n\}\}}} \mathbb{E}^{P_i^o}\Bigg\{\sum_{i=t}^n\log\left( \frac{P_i(Y_i|Y_{i-1},X_i)}{P_i(Y_i|Y_{i-1})}\right) \\
    &- \sum_{i=t}^n s_i\bigg(\rho_i(X_i,Y_i)-D_i[y_{i-1},P_i^o]\bigg)\Bigg|{Y_{i-1}=y_{i-1}}\Bigg\}.
\end{align*}
Applying the \textit{principle of optimality} \cite{bertsekas:2005} yields the following finite-time horizon stochastic DP recursions obtained backward in time
\begin{align}
&R_n(D_n[y_{n-1},P_n^o]) =\min_{P_n(y_n|y_{n-1},x_n)}\sum_{x_{n-1}^n\in\mathcal{X}_{n-1}^n,y_n\in\mathcal{Y}_n} \nonumber\\
&\Bigg\{\Bigg(\log\left( \frac{P_n(y_n|y_{n-1},x_n)}{P_n(y_n|y_{n-1})}\right)- s_n \rho_n(x_n,y_n)\Bigg) P_n(x_n|x_{n-1})\nonumber
\end{align}
\begin{equation}
P_n(y_n|y_{n-1},x_n)P_n^o(x_{n-1}|y_{n-1}) \Bigg\}+s_n D_n[y_{n-1}, P_n^o]\label{DP-time-n}
\end{equation}
\begin{align}
    &R_t(D_t[y_{t-1},P_t^o])=\min_{P_t(y_t|y_{t-1},x_t)} \sum_{x_{t-1}^t\in\mathcal{X}_{t-1}^t,y_t\in\mathcal{Y}_t}\nonumber\\
    & \Bigg\{\Bigg(\log\left( \frac{P_t(y_t|y_{t-1},x_t)}{P_t(y_t|y_{t-1})}\right)- s_t \rho_t(x_t,y_t) \nonumber\\
    &+R_{t+1}(D_{t+1}[y_t,P_{t+1}^o])\Bigg)P_t(x_t|x_{t-1})P_t(y_t|y_{t-1},x_t)\nonumber\\
    &P_t^o(x_{t-1}|y_{t-1}) \Bigg\}+s_t D_t[y_{t-1},P_t^o], ~t=n-1,n-2,\ldots,0.
\label{DP-time-t}
\end{align}
Note that the given information or belief state for fixed $Y_{t-1}=y_{t-1}$ in \eqref{DP-time-n}, \eqref{DP-time-t} can be identified at each time-stage $t$ by the following recursion
\begin{align}
&P_{t+1}^o(x_t|y_t)=\frac{P_t(x_t, y_t)}{P_t(y_t)}\nonumber\\
&=\frac{\sum_{\substack{y_{t-1}\in{\cal Y}_{t-1}\\ x_{t-1}\in{\cal X}_{t-1}}} P_t(y_t|y_{t-1},x_t) P_t(x_t|x_{t-1}) P^o_{t}(x_{t-1}|y_{t-1})}{\sum_{\substack{y_{t-1}\in{\cal Y}_{t-1}\\ x_{t-1}^t\in{\cal X}_{t-1}^t}} P_t(y_t|y_{t-1},x_t) P_t(x_t|x_{t-1}) P^o_{t}(x_{t-1}|y_{t-1})}\nonumber
\end{align}
which is Markov, conditional on $Y_t, P^o_{t}(x_{t-1}|y_{t-1})$. Moreover, at $t=0$ we adopt the assumptions of our system model in Fig. \ref{fig:l-sc} and assume that $P_0(y_0|y_{-1},x_0)=P_0(y_0|x_0)$, and $P_0(y_0|y_{-1})=P_0(y_0)$ hence the initial time stage in \eqref{DP-time-t} can be modified accordingly. To approximate \eqref{NRDF-simplified} for a given $\{P_t(x_{t-1}|y_{t-1})\equiv P_t^o(x_{t-1}|y_{t-1}):~t\in\mathbb{N}_1^n\}$, we need to move forward in time (online computation) starting from $R_0$ and obtain the clean values of the cost-to-go at each time stage averaging over $y_{t-1}\in{\mathcal{Y}}_{t-1},~\forall{t}\in\mathbb{N}_1^n$ and exploring over all possible values of the belief (information) state space generated by $\{P_t(x_{t-1}|y_{t-1})\equiv P_t^o(x_{t-1}|y_{t-1}):~t\in\mathbb{N}_1^n\}$. 

The partially observable stochastic DP algorithm in \eqref{DP-time-n}, \eqref{DP-time-t} can be solved offline using approximation methods such as discretizing the continuous belief state space \cite{bertsekas:2005}. 
However, instead of following that direction, we can leverage the convexity of our problem and optimize wrt the control policy (test-channel) in order to find a more efficient way to compute the DP recursions via policy evaluation. Indeed, if implicit analytical recursions of the control policy are available, this will enable faster computation of both the cost and the control policy \cite{bertsekas:2005}.

To do it, we first discretize the given continuous belief state space into a finite state space which we denote by $\mathcal{B}_t,\forall t\in\mathbb{N}_0^n$. Then, we rely on a dynamic version of an  \textit{AM scheme} the way it was used to compute the classical RDF for discrete sources \cite{blahut:1982} to generate an offline training algorithm. To implement the dynamic AM scheme, we need the following lemma. 

\begin{lemma}{(Double minimization)}
\label{lemma:double-min}
For any $t\in\mathbb{N}_0^n$, let $s_t\leq0$ and $D_t>0$, then for a fixed Markov source $P_t(x_t|x_{t-1})$,  and a given belief state $\{P_t(x_{t-1}|y_{t-1})=P_t^o(x_{t-1}|y_{t-1})\}\in\mathcal{B}_t$ obtained for a fixed $Y_{t-1}=y_{t-1}$, (\ref{DP-time-n}), (\ref{DP-time-t}) can be expressed as a double minimum as follows
\begin{align}
    &R_t(D_t[y_{t-1},P_t^o]) = \min_{P_t(y_t|y_{t-1},x_t)}\min_{P_t(y_t|y_{t-1})}\sum_{x_{t-1}^t\in\mathcal{X}_{t-1}^t,y_t\in\mathcal{Y}_t}\nonumber\\
    &\Bigg\{\Bigg(\log\left( \frac{P_t(y_t|y_{t-1},x_t)}{P_t(y_t|y_{t-1})}\right)-s_t \rho_t(x_t,y_t)\nonumber\\
    & +R_{t+1}(D_{t+1}[y_t,P_{t+1}^o])\Bigg)
    P_t(y_t|y_{t-1},x_t)P_t(x_t|x_{t-1})\nonumber\\ &P_t^o(x_{t-1}|y_{t-1}) \Bigg\} + s_t D_t[y_{t-1},P_t^o], ~t=n,n-1,\ldots,0
\label{DM-CostToGo}
\end{align}
where $R_{t+1}(D_{t+1}[y_t,P_{t+1}^o])$ is the cost-to-go that is equal to $0$ when $t=n$, and $D_t[y_{t-1},P_t^o]$ is a functional of the distortion level expressed as
\begin{align}
    &D_t[y_{t-1},P_t^o]   \nonumber\\ &=\sum_{\substack{x_t\in\mathcal{X}_t\\y_t\in\mathcal{Y}_t}}P_t^{\text{M}o}(x_t|y_{t-1})P_t^*(y_t|y_{t-1},x_t)\rho_t(x_t,y_t)
\label{distortion-stagewise}
\end{align}
with $P_t^{\text{M}o}(x_t|y_{t-1})=\sum_{x_{t-1}\in\mathcal{X}_{t-1}}P_t(x_t|x_{t-1})P_t^o(x_{t-1}|y_{t-1})$ and $P_t^*(y_t|y_{t-1},x_t)$ is achieving the minimum. Moreover, for a fixed $P_t(y_t|y_{t-1},x_t)$, the right hand side (RHS) of \eqref{DM-CostToGo} is minimized by 
\begin{align}
    P_t(y_t|y_{t-1}) = \sum_{x_t\in\mathcal{X}_t} P_t^{\text{M}o}(x_t|y_{t-1}) P_t(y_t|y_{t-1},x_{t})
\label{AM-output}
\end{align}
whereas for fixed $P_t(y_t|y_{t-1})$, the RHS of (\ref{DM-CostToGo}) is minimized by 
\begin{align}
    P_t(y_t|y_{t-1},x_t) = \frac{P_t(y_t|y_{t-1})A_t[P_{t+1}^o](x_t,y_t,s_t)}{\sum_{y_t\in\mathcal{Y}_t}P_t(y_t|y_{t-1})A_t[P_{t+1}^o](x_t,y_t,s_t)}
\label{AM-policy}
\end{align}
where $A_t[P_{t+1}^o](x_t,y_t,s_t) = e^{s_t\rho_t(x_t,y_t)-R_{t+1}(D_{t+1}[y_t,P_{t+1}^o])}$.
\end{lemma}
\begin{proof}
We outline the proof due to space limitations. The double minimization at each time stage in \eqref{DM-CostToGo}, follows immediately using the same arguments as in \cite[Theorem 5.2.6]{blahut:1982}. Since the problem is convex, or more generally biconvex if we fix $P_t(y_t|y_{t-1})$ and optimize wrt to $P_t(y_t|y_{t-1},x_t)$ (and vice-versa), we can apply KKT conditions \cite{boyd:2004} at each time stage in \eqref{DM-CostToGo} and obtain \eqref{AM-output} and \eqref{AM-policy}, respectively.
\end{proof}

Clearly, from Lemma \ref{lemma:double-min}, if we substitute \eqref{AM-policy} in \eqref{distortion-stagewise}, we obtain
\begin{align}
    &D_{s_t}[y_{t-1},P_t^o]= \sum_{x_t\in\mathcal{X}_t,y_t\in\mathcal{Y}_t}P_t^{\text{M}o}(x_t|y_{t-1}) \nonumber\\
    &\frac{P_t^*(y_t|y_{t-1})A_t[P_{t+1}^o](x_t,y_t,s_t)}{\sum_{y_t\in\mathcal{Y}_t}P_t^*(y_t|y_{t-1})A_t[P_{t+1}^o](x_t,y_t,s_t)} \rho_t(x_t,y_t).
\label{distortion-point-s}
\end{align}

Lemma \ref{lemma:double-min} provides the parametric model that enables us to utilize a new dynamic version of BAA \cite{blahut:1982} to construct an AM scheme between $P_t(y_t|y_{t-1},x_t)$ and $P_t(y_t|y_{t-1})$ at any $t\in\mathbb{N}_0^n$.
In the sequel, we describe mathematically the convergence of an offline training algorithm that leads to the approximate computation of the optimal control policy by a dynamic AM.

\begin{theorem}{(Offline training algorithm)}
\label{theorem:convergence}
For each $t\in\mathbb{N}_0^n$, consider a fixed source $P_t(x_t|x_{t-1})$, and a given $P_t(x_{t-1}|y_{t-1})\equiv P_t^o(x_{t-1}|y_{t-1})$ obtained for a fixed $Y_{t-1}=y_{t-1}$.
Moreover, for any $t$, let $s_t\leq0$ and $P_t^{(0)}(y_t|y_{t-1})$ be the initial output probability distribution with non-zero components, and let $P_t^{(k+1)}(y_t|y_{t-1})=P_t[P_t^{(k)}(y_t|y_{t-1})](y_t|y_{t-1},x_t)$ and $P_t^{(k+1)}(y_t|y_{t-1},x_t)=P_t[P_t^{(k)}(y_t|y_{t-1})](y_t|y_{t-1})$ be expressed as follows
\begin{align}
    &P_t^{(k+1)}(y_t|y_{t-1},x_t) \nonumber\\
    &= P_t^{(k)}(y_t|y_{t-1})  \frac{A_t[P_{t+1}^o](x_t,y_t,s_t)}{\sum_{y_t\in\mathcal{Y}_t}P_t^{(k)}(y_t|y_{t-1})A_t[P_{t+1}^o](x_t,y_t,s_t)}
    \label{policy-update}
    \\
    &P_t^{(k+1)}(y_t|y_{t-1}) = P_t^{(k)}(y_t|y_{t-1}) \nonumber\\
    &\quad\sum_{x_t\in\mathcal{X}_t}  \frac{P_t^{\text{M}o}(x_t|y_{t-1})A_t[P_{t+1}^o](x_t,y_t,s_t)}{\sum_{y_t\in\mathcal{Y}_t}P_t^{(k)}(y_t|y_{t-1})A_t[P_{t+1}^o](x_t,y_t,s_t)}.
    \label{output-update}
\end{align}
Then as $k\to\infty$, we obtain for any $t$ that 
\begin{align*}
    D_t[y_{t-1},P_t^o,P_t^{(k)}(y_t|y_{t-1},x_t)] &\to D_{s_t}[y_{t-1},P_t^o] \nonumber\\
    \mathbb{I}_t[y_{t-1},P_t^o,P_t^{(k)}(y_t|y_{t-1},x_t)]&\to R_t(D_{s_t}[y_{t-1},P_t^o])
\end{align*}
where $(D_{s_t}[y_{t-1},P_t^o],R_t(D_{s_t}[y_{t-1},P_t^o]))$ denotes a point on the rate-distortion curve parametrized by $s_t$, given $P_t^o(x_{t-1}|y_{t-1})$ obtained for a fixed $Y_{t-1}=y_{t-1}$ whereas $\mathbb{I}_t[y_{t-1},P_t^o,P_t^{(k)}(y_t|y_{t-1},x_t)]$ and $D_t[y_{t-1},P_t^o,P_t^{(k)}(y_t|y_{t-1},x_t)]$ are expressed as
\begin{align}
    &\mathbb{I}_t[y_{t-1},P_t^o,P_t^{(k)}(y_t|y_{t-1},x_t)] = \sum_{x_t\in\mathcal{X}_t,y_t\in\mathcal{Y}_t}P_t^{\text{M}o}(x_t|y_{t-1})\nonumber\\
    &\qquad P_t^{(k)}(y_t|y_{t-1},x_t) 
    \bigg(\log\bigg(\frac{P_t^{(k)}(y_t|y_{t-1},x_t)}{P_t^{(k)}(y_t|y_{t-1})}\bigg)\nonumber\\&\qquad+R_{t+1}(D_{t+1}[y_t,P_{t+1}^o])\bigg)
\label{CostToGo-belief-policy}\\
    &D_t[y_{t-1},P_t^o,P_t^{(k)}(y_t|y_{t-1},x_t)]\nonumber\\&\quad=\sum_{x_t\in\mathcal{X}_t,y_t\in\mathcal{Y}_t}P_t^{\text{M}o}(x_t|y_{t-1})P_t^{(k)}(y_t|y_{t-1},x_t) \rho_t(x_t,y_t).
\label{Distortion-belief-policy}
\end{align}
\end{theorem}
\begin{proof}
We sketch the proof due to the space limitations. We first demonstrate that $\mathbb{I}_t[y_{t-1},P_t^o]-D_t[y_{t-1},P_t^o]$ iteratively updated by $P_t^{(k)}(y_t|y_{t-1})$ forms a non-increasing sequence under alternating minimization. Since this sequence is also bounded from below, it converges to a finite limit. Next, we establish a lower bound on the difference between consecutive updates of $\mathbb{I}_t[y_{t-1},P_t^o]-D_t[y_{t-1},P_t^o]$ by \eqref{policy-update} and \eqref{output-update}, and show that this difference decreases across iterations, which ensures the convergence.
\end{proof}
The implementation of Theorem \ref{theorem:convergence} is illustrated in Algorithm \ref{algo:DP-backward}. The following theorem supplements Algorithm \ref{algo:DP-backward} with a stopping criterion after a finite number of steps.
\begin{theorem}
\label{theorem:stopping-criterion}
(Stopping criterion of Algorithm \ref{algo:DP-backward}) 
For each $t\in\mathbb{N}_0^n$, the point $D_{s_t}[y_{t-1},P_t^o]$ given by \eqref{distortion-point-s} admits the following bounds
\begin{align}
    R_t&(D_{s_t}[y_{t-1},P_t^o]) \leq s_t D_{s_t}[y_{t-1},P_t^o] -\sum_{x_t\in\mathcal{X}_t}\bigg(P_t^{\text{M}o}(x_t|y_{t-1}) \nonumber\\
    &\log(\sum_{y_t\in\mathcal{Y}_t}P_t(y_t|y_{t-1})A_t[P_{t+1}^o](x_t,y_t,s_t))\bigg) \nonumber\\ 
    &-\sum_{y_t\in\mathcal{Y}_t}P_t(y_t|y_{t-1})c_t[y_{t-1}](y_t)\log c_t[y_{t-1}](y_t),
\label{CostToGo-upper-bound}\\
    R_t&(D_{s_t}[y_{t-1},P_t^o])\geq s_t D_{s_t}[y_{t-1},P_t^o]-\sum_{x_t\in\mathcal{X}_t}\bigg(P_t^{\text{M}o}(x_t|y_{t-1})\nonumber\\
    &\log(\sum_{y_t\in\mathcal{Y}_t}P_t(y_t|y_{t-1})A_t[P_{t+1}^o](x_t,y_t,s_t))\bigg)\nonumber\\ 
    &-\max_{y_t\in\mathcal{Y}_t}\log c_t[y_{t-1}](y_t),
\label{CostToGo-lower-bound}
\end{align}
where $c_t[y_{t-1}](y_t)$ is expressed as a function of fixed $y_{t-1}$ 
\begin{align*}
    c_t[y_{t-1}](y_t) = &\sum_{x_t\in\mathcal{X}_t} P_t^{\text{M}o}(x_t|y_{t-1})\nonumber\\& \frac{A_t[P_{t+1}^o](x_t,y_t,s_t)}{\sum_{y_t\in\mathcal{Y}_t}P_t(y_t|y_{t-1})A_t[P_{t+1}^o](x_t,y_t,s_t)}.
\end{align*}
\end{theorem}
\begin{proof}
The proof follows if in both \eqref{CostToGo-upper-bound}, \eqref{CostToGo-lower-bound} we use backward induction following at each time instant the derivation of \cite[Theorem 6.3.10]{blahut:1982}.
\end{proof}
Theorem \ref{theorem:stopping-criterion}, generates a stopping criterion for Algorithm \ref{algo:DP-backward} at the $k$-th iteration by setting the estimation error per time stage, i.e., $T_{U_t}[y_{t-1},P_t^o]-T_{L_t}[y_{t-1},P_t^o]$ where
\begin{align*}
    T_{U_t}[y_{t-1},P_t^o] &= \sum_{y_t\in\mathcal{Y}_t}P_t(y_t|y_{t-1})c_t[y_{t-1}](y_t)\log c_t[y_{t-1}](y_t)\nonumber\\
    T_{L_t}[y_{t-1},P_t^o] &= \max_{y_t\in\mathcal{Y}_t}\log c_t[y_{t-1}](y_t).
\end{align*}


\begin{algorithm}
\caption{Approximation of the Control Policy Backward in Time (Offline Training)}
\label{algo:DP-backward}
\begin{algorithmic}[1]
    \REQUIRE given source distribution$\{P_t(x_t|x_{t-1}):t\in\mathbb{N}_0^n\}$,\\ given belief state $P_t^o(x_{t-1}|y_{t-1})\in\mathcal{B}_t$, Lagrange multipliers $\{s_t\leq0:t\in\mathbb{N}_0^n\}$, error tolerance $\epsilon>0$
    \STATE {\textbf{Initialize} $\{P_t^{(0)}(y_t|y_{t-1}):t\in\mathbb{N}_0^n\}$\;}\\
    \FOR{$t=n:1$}
    \STATE $k\leftarrow0$
    \WHILE {$T_{U_t}[y_{t-1},P_t^o] - T_{L_t}[y_{t-1},P_t^o]  > \epsilon$}
    \STATE $P_t^{(k)}(y_t|y_{t-1}, x_t)\leftarrow$ (\ref{policy-update}) \;
    \STATE $P_t^{(k+1)}(y_t|y_{t-1})\leftarrow$ (\ref{output-update})\;
    \STATE $R_t(D_t[y_{t-1},P_t^o]) \leftarrow$~\eqref{CostToGo-belief-policy} 
    \STATE $k \leftarrow k + 1$\;
    \ENDWHILE
    \ENDFOR
    \ENSURE $\{P_t^*[P_t^o](y_t|y_{t-1}, x_t):t\in\mathbb{N}_1^n\}$,\\ $\{P_t^*[P_t^o](y_t|y_{t-1}):t\in\mathbb{N}_1^n\}$,\\ $\{R_t(D_{s_t}[y_{t-1},P_t^o]):t\in\mathbb{N}_1^n\}$ 
\end{algorithmic}
\end{algorithm}
\paragraph*{Comments on Algorithms \ref{algo:DP-backward}, \ref{algo:DP-forward}} Algorithm \ref{algo:DP-backward} approximates the  control policy (test-channel) $\{P_t^*[P_t^o](y_t|y_{t-1}, x_t):t\in\mathbb{N}_1^n\}$, the output distribution $\{P_t^*[P_t^o](y_t|y_{t-1}):t\in\mathbb{N}_1^n\}$, and the cost-to-go function $\{R_t(D_{s_t}[y_{t-1},P_t^o]):t\in\mathbb{N}_1^n\}$ 
as functions of the fixed $Y_{t-1}=y_{t-1}$ and the discretized belief state space $P_t^o(x_{t-1}|y_{t-1})\in\mathcal{B}_t$, and also the one-step look-ahead (time stage $t+1$) quantized belief state space $P_{t+1}^o(x_t|y_t)\in\mathcal{B}_{t+1}$. After collecting these functions backward from $t=n$ to $t=1$, we initiate the online Algorithm \ref{algo:DP-forward} to evaluate the cost-to-go forward in time and identify the optimizing distributions that can approximate the minimum in \eqref{NRDF-simplified}. The source distribution $P_0(x_0)$ and output distribution $P_0(y_0)$ at $t=0$ are given as initial conditions for forward computation, from which the initial control policy can be obtained $P_0(y_0|x_0)$, hence the initial posterior $P_1(x_0|y_0)$, which serves as the initial belief state $P_1^*(x_0|y_0)$. For each $t$, the best policies $P_t^*(y_t|y_{t-1},x_t)$ are determined by following the best trajectory $P_{t+1}^*(x_t|y_t)$ such that
\begin{align}
    &P_{t+1}^*(x_t|y_t)    =\arg\min_{\substack{P_{t+1}^o(x_t|y_t) \in\mathcal{B}_{t+1}}} \nonumber\\
    &\sum_{y_{t-1}\in\mathcal{Y}_{t-1}} R_t(D_{s_t}[y_{t-1},P_t^*])P_t(y_{t-1}),~\forall t=\mathbb{N}_2^{n-1}
\label{belief-argmin}
\end{align}
and eventually the minimum in (\ref{NRDF-simplified}) is approximated. Clearly, the larger the search space of the finite belief state, the better the approximation. Ideally, a sufficiently large belief state space can approximate near-optimally the minimum in (\ref{NRDF-simplified}).

\begin{algorithm}
\caption{Forward Computation of the Approximate Control Policy (Online Computation)}
\label{algo:DP-forward}
\begin{algorithmic}[1]
    \REQUIRE {$\{\mathcal{B}_t:t\in\mathbb{N}_1^n\}$ of given $\{P_t^o(x_{t-1}|y_{t-1}):t\in\mathbb{N}_1^n\}$,\\
    $\{P_t^*[P_t^o](y_t|y_{t-1}, x_t):t\in\mathbb{N}_1^n\}$, \\$\{P_t^*[P_t^o](y_t|y_{t-1}):t\in\mathbb{N}_1^n\}$,\\ $\{R_t(D_{s_t}[y_{t-1},P_t^o]):t\in\mathbb{N}_1^n\}$.}
    \STATE \textbf{Initialize} $P_0(x_0)$, $P_0(y_0)$, $P_1^*(x_0|y_0)=P(x_0|y_0)$ \;\\
    \FOR {$t = 1:n-1$}
    \STATE $P_{t+1}^*(x_t|y_t) \leftarrow$~\eqref{belief-argmin} \;
    \STATE $P_t^*(y_t|y_{t-1},x_t)\leftarrow$ \\ $\qquad P_t^*[P_t^*(x_{t-1}|y_{t-1}),P_{t+1}^*(x_t|y_t)](y_t|y_{t-1},x_t)$ \;
    \STATE $P_t^*(y_t|y_{t-1})\leftarrow$ \\ $\qquad P_t^*[P_t^*(x_{t-1}|y_{t-1}),P_{t+1}^*(x_t|y_t)](y_t|y_{t-1})$ \;
    \ENDFOR
    \STATE $P_n^*(y_n|y_{n-1},x_n)\leftarrow P_n^*[P_n^*(x_{n-1}|y_{n-1})](y_n|y_{t-1},x_n)$\;
    \STATE $P_n^*(y_n|y_{n-1})\leftarrow P_n^*[P_n^*(x_{n-1}|y_{n-1})](y_n|y_{n-1})$\;
    \ENSURE {$\quad$\\ $\{P_t^*(x_{t-1}|y_{t-1}):t\in\mathbb{N}_1^n\}$, $\{P_t^*(y_t|y_{t-1},x_t):t\in\mathbb{N}_0^n\}$,\\ $\{P_t^*(y_t|y_{t-1}):t\in\mathbb{N}_0^n\}$, $R_{[0,n]}^{na}(D_0,D_1,\ldots,D_n)=\sum_{t=0}^n I[P_t^*(x_{t-1}|y_{t-1}),P_t^*(y_t|y_{t-1},x_t)](X_t;Y_t|Y_{t-1})$.}
\end{algorithmic}
\end{algorithm}

\section{Numerical Examples}
\label{numerical-results}
This section provides numerical simulations to support our theoretical findings that led to Algorithms \ref{algo:DP-backward}, \ref{algo:DP-forward}. We consider two examples, both assuming binary alphabet spaces $\{\mathcal{X}_t=\mathcal{Y}_t=\{0,1\}:t\in\mathbb{N}_0^n\}$, with Hamming distortion given by
\begin{equation}
    \rho_t(x_t,y_t)\equiv\rho(x_t,y_t) = \bigg\{
    \begin{matrix}
        0, \ \ \text{if}\ x_t= y_t \\
        1, \ \ \text{if}\ x_t\neq y_t
    \end{matrix}, ~~\forall{t}.
\end{equation}
Moreover, we consider a belief state $P_t^o(x_{t-1}|y_{t-1})\in\mathcal{B}_t$, that consists of a matrix comprising two ``quantized'' binary probability distributions drawn from the original continuous state space.  We denote with $N_t$ each quantization level per $t$, which leads to a belief state space $\mathcal{B}_t$ with size $|\mathcal{B}_t| = N_t^2$, representing combinations of $2$ out of $N_t$ quantized binary distributions.

\begin{table}
\centering
\caption{Time consumption comparison}
\begin{tabular}{crrrrr}
    \hline
    Core & $N$=10 & $N$=15 & $N$=20 & $N$=25 & $N$=30 \\ \hline
    1 & 11:59.20 & 58:21.99 & 3:10:02.22 & 7:28:14.57 & 16:01:47.46 \\
    8 & 2:14.15 & 10:38.16 & 32:55.82 & 1:21:00.55 & 2:47:20.54 \\
    16 & 1:09.80 & 05:59.45 & 19:18.00 & 47:07.67 & 1:38:51.39 \\
    \hline
\end{tabular}
\label{table:time-consumption}
\end{table}

\begin{xmpl}
\label{example-time-varying}
(Time-varying binary symmetric Markov source)~
 The source distribution $P_t(x_t|x_{t-1})$ at each $t\in\mathbb{N}_1^n$ is chosen such that for each $t$, we have
\begin{equation}
\begin{split}
    P_t(x_t|x_{t-1}) = 
    \begin{pmatrix}
        1-\alpha_t & \alpha_t \\
        \alpha_t & 1-\alpha_t 
    \end{pmatrix}, ~\alpha_t\in(0,1).
\end{split}
\end{equation}
Moreover, we choose the quantization levels $\{N_t=N:t\in\mathbb{N}_1^n\}$ and the stagewise Lagrangian $\{s_t=s:t\in\mathbb{N}_0^n\}$. In Fig. \ref{fig:numerical-TV-RD-convergence}, we demonstrate some results applying Algorithms \ref{algo:DP-backward}, \eqref{algo:DP-forward} for $N=20$, $s=-2$, and $n=100$, whereas in Fig. \ref{fig-convergence} we illustrate several time stages selected during backward computation to verify the convergence of Algorithm \ref{algo:DP-backward}. 
For our example, in Table \ref{table:time-consumption} we compare the time consumption of Algorithm \ref{algo:DP-backward} across single and multi-core processing for various $N$ over a time horizon of $n=100$. Our results demonstrate significant improvement in terms of the computational time complexity with multi-core processing compared to the single-thread processing.
\end{xmpl}
\begin{figure}
\centering
\begin{subfigure}[b]{0.49\linewidth}
    \centering
    \includegraphics[height=3.6cm,width=\linewidth]{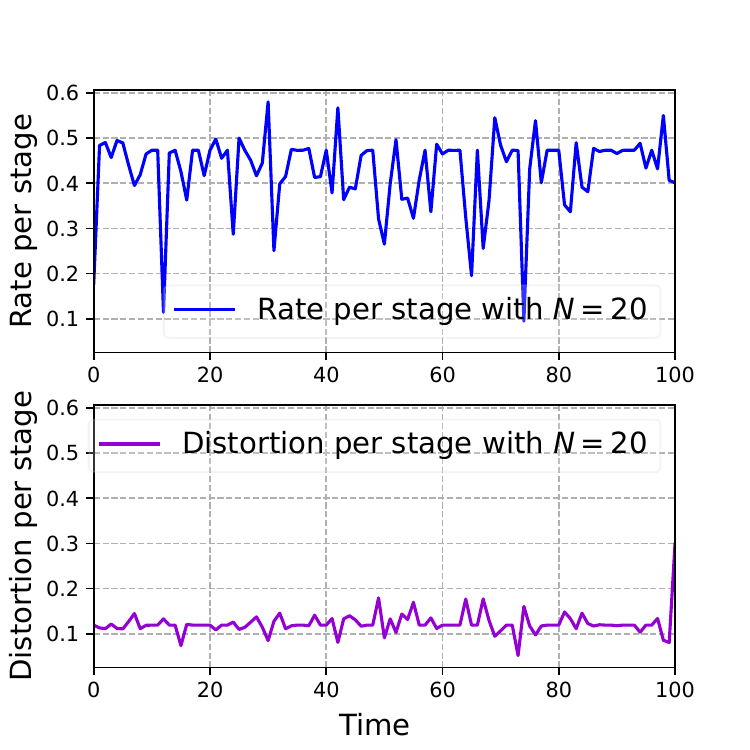}
    \caption{Rate \& distortion per stage}
    \label{fig-RD-TV}
\end{subfigure}
\begin{subfigure}[b]{0.49\linewidth}
    \centering
    \includegraphics[height=3.6cm,width=\linewidth]{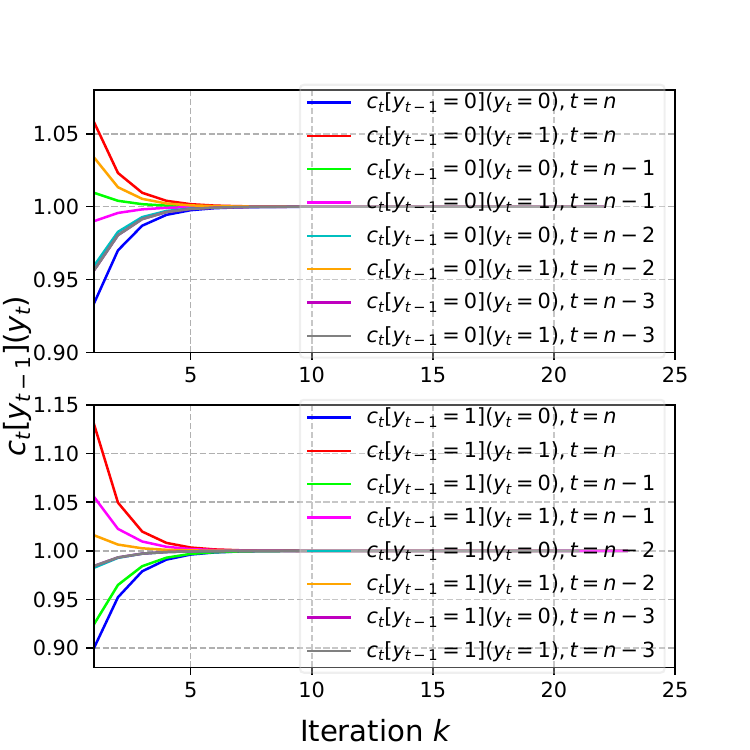}
    \caption{Convergence per stage}
    \label{fig-convergence}
\end{subfigure}
\caption{Illustration of the rate \& distortion per stage and stage-wise convergence for the time-varying case.}
\label{fig:numerical-TV-RD-convergence}
\end{figure}

\begin{xmpl}
\label{example-time-invariant}
(Time-invariant binary symmetric Markov source)
The source distribution $P_t(x_t|x_{t-1})\equiv{P}(x_t|x_{t-1}),~\forall{t}$, where
\begin{equation}
\begin{split}
	P(x_t|x_{t-1}) = 
	\begin{pmatrix}
		1-\alpha & \alpha \\
		\alpha & 1-\alpha 
	\end{pmatrix},~\alpha\in(0,1)~\forall{t}.
\end{split}
\end{equation}
In Fig. \ref{fig:numerical-invariant-N30} we show results for $\alpha=0.4$, $N=30$, $s=-2$, and $n=100$. Specifically, in Figs. \ref{fig-belief-30}-\ref{fig-testchannel-30}, we observe that the transients of the belief state, the output distribution and the control policy, respectively, reach a stationary value apart from the initial and terminal stages and this is also reflected in the behavior of the rate and distortion per stage in Fig. \ref{fig-rate-30}. Interestingly, when the belief state, the output distribution and the control policy obtain stationary distributions, these maintain a symmetric structure. 
\end{xmpl}

\begin{figure}
\centering
\begin{subfigure}[b]{0.49\linewidth}
    \centering
\includegraphics[height=3.6cm,width=\linewidth]{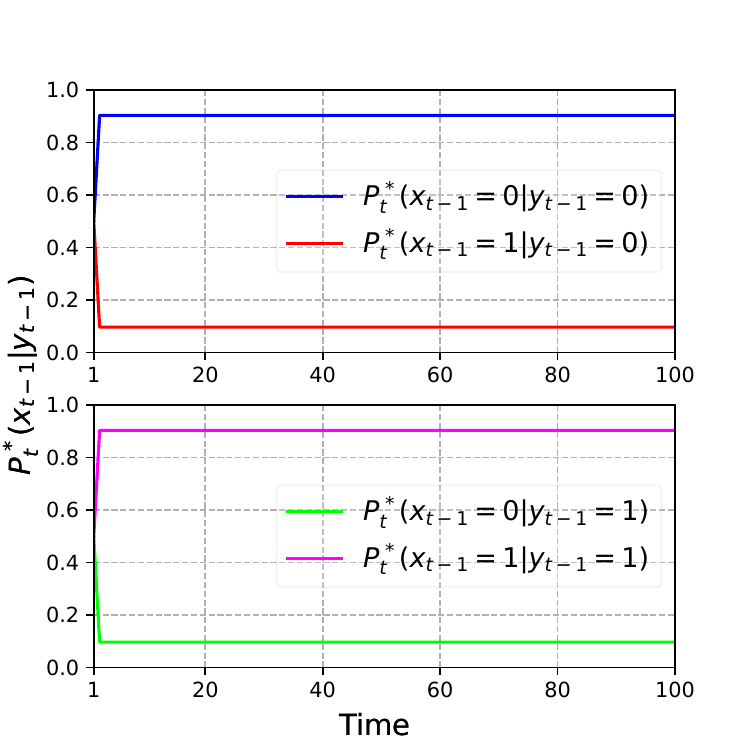}
    \caption{$P_t^*(x_{t-1}|y_{t-1})$}
    \label{fig-belief-30}
\end{subfigure}
\begin{subfigure}[b]{0.49\linewidth}
    \centering
    \includegraphics[height=3.6cm,width=\linewidth]{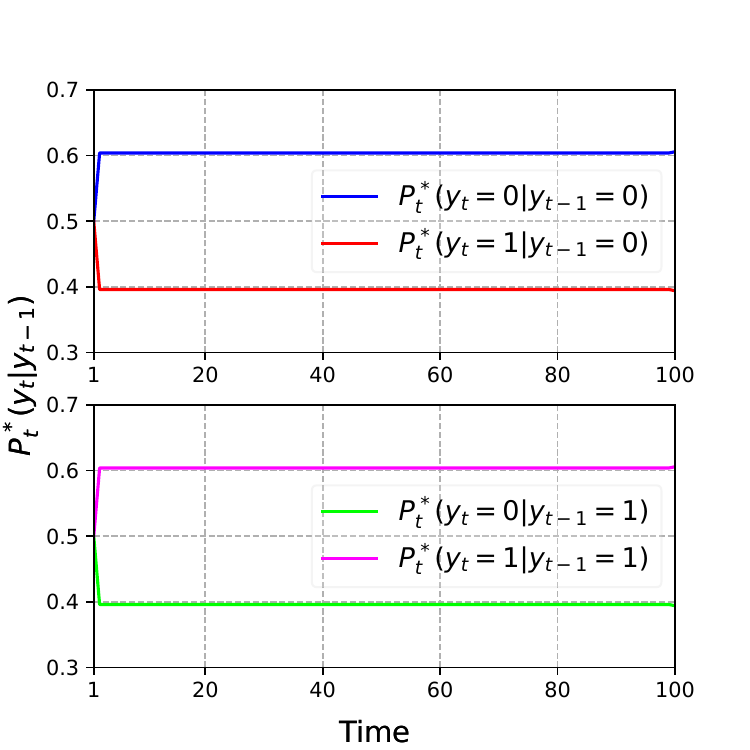}
    \caption{$P_t^*(y_t|y_{t-1})$}
    \label{fig-output-30}
\end{subfigure}
\begin{subfigure}[b]{0.49\linewidth}
    \centering
    \includegraphics[height=3.6cm,width=\linewidth]{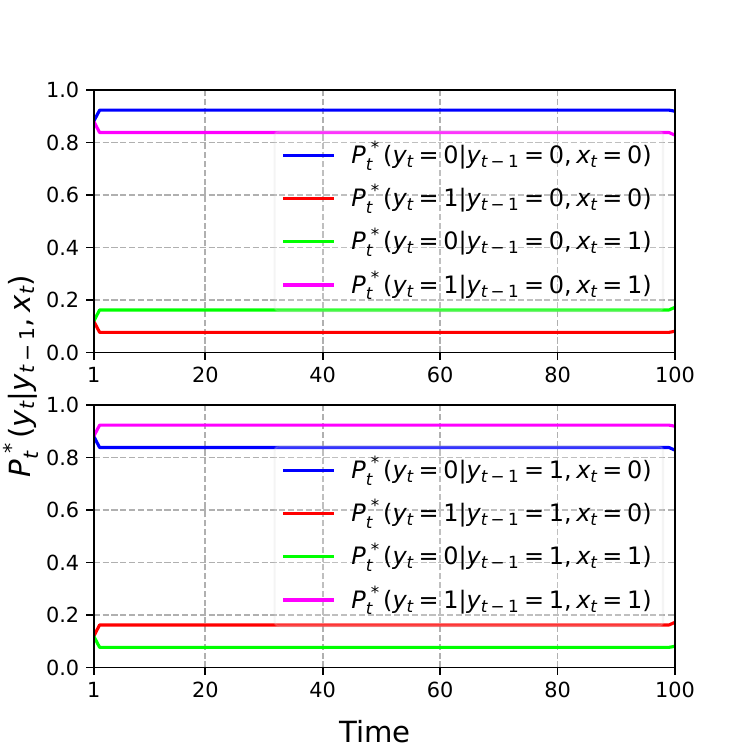}
    \caption{$P_t^*(y_t|y_{t-1},x_t)$}
    \label{fig-testchannel-30}
\end{subfigure}
\begin{subfigure}[b]{0.49\linewidth}
    \centering
    \includegraphics[height=3.6cm,width=\linewidth]{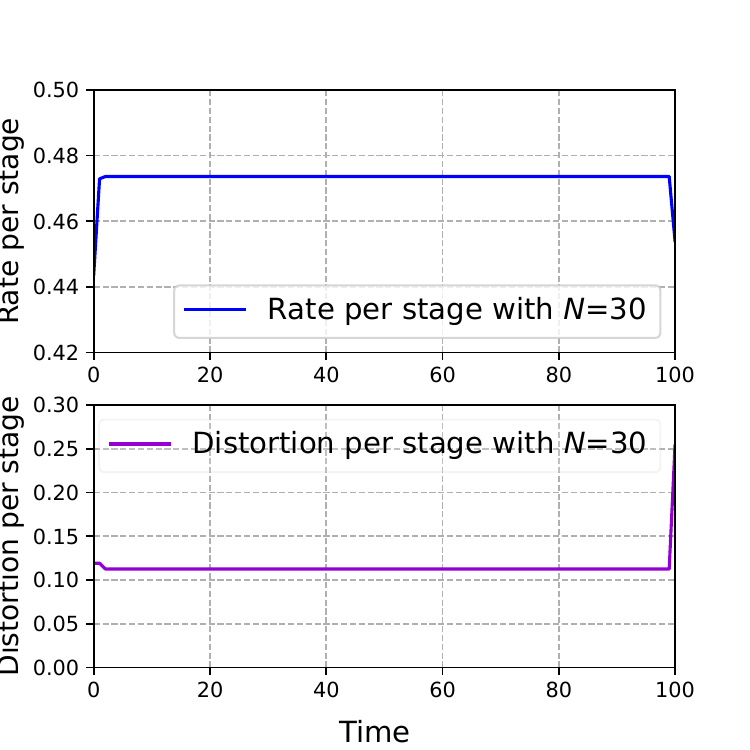}
    \caption{Rate \& distortion per stage}
    \label{fig-rate-30}
\end{subfigure}
\caption{Illustration of the best trajectory $\{P_t^*(x_{t-1}|y_{t-1}):t\in\mathbb{N}_1^n\}$, the corresponding test-channels, output, and rate per stage ($t\in\mathbb{N}_0^n$) in the time-invariant case.}
\label{fig:numerical-invariant-N30}
\end{figure}
\section{Conclusion}
\label{conclusion}
We derived a non-asymptotic lower bound for a zero-delay variable-rate lossy source coding system assuming discrete Markov sources. By leveraging some new structural and convexity properties of NRDF, we approximated the problem via an unconstrained partially observable finite-horizon stochastic DP and proposed a novel dynamic AM scheme to compute the control policy and the cost-to-go function through an offline training algorithm followed by an online computation. Our theoretical results are supplemented with simulation studies that considered binary Markov processes.

\bibliographystyle{IEEEtran}
\bibliography{string,literature_conf}

\end{document}